\documentclass[8pt,draftcls,twocolumn]{autart}

\usepackage{graphicx}          
\usepackage{amsfonts}
\usepackage{epsfig}
\usepackage{subfigure}
\usepackage{amsmath,color,mathtools,enumerate}
\usepackage[all]{xy}
\usepackage{float}
\usepackage{mathrsfs}
\usepackage{amssymb}
\usepackage{CJK}
\newtheorem{definition}{Definition}[section]
\newtheorem{lemma}{Lemma}[section]
\newtheorem{theorem}{Theorem}[section]

\newtheorem{remark}{Remark}[section]
\newtheorem{assumption}{Assumption}[section]
\newtheorem{proposition}{Proposition}[section]

\def\d{ {\rm d  }}

\newcommand{\citep}{\cite}
\newcommand{\citet}{\cite}

\begin{document}
\begin{frontmatter}
\title{Coordinated trajectory tracking of multiple vertical take-off and landing UAVs \thanksref{footnoteinfo}}

\thanks[footnoteinfo]{This paper was not presented at any IFAC
meeting. This work has been supported in part
by National
Natural Science Foundation of China under Grant 61503249 and Beijing Municipal Natural Science Foundation under Grant 4173075, and the National Key Research and Development Program of China under Grant 2016YFB0500900/2.  \\Corresponding author: Ziyang Meng.}
\author[Tsinghua]{Yao~Zou}\ead{zouyao@mail.tsinghua.edu.cn},
\author[Tsinghua]{Ziyang~Meng}\ead{ziyangmeng@mail.tsinghua.edu.cn},

\address[Tsinghua]{Department of Precision Instrument, Tsinghua University, Beijing 100084, P. R. China.}


\begin{abstract}
This paper investigates the coordinated trajectory tracking problem of multiple vertical takeooff and landing (VTOL) unmanned aerial vehicles (UAVs). The case of unidirectional information flow is considered and the objective is to drive all the follower VTOL UAVs to accurately track the trajectory of the leader. Firstly, a novel distributed estimator is developed for each VTOL UAV to obtain the leader's desired information asymptotically. With the outputs of the estimators, the solution to the coordinated trajectory tracking problem of multiple VTOL UAVs is transformed to individually solving the tracking problem of each VTOL UAV. Due to the under-actuated nature of the VTOL UAV, a hierarchical framework is introduced for each VTOL UAV such that a command force and an applied torque are exploited in sequence, then the position tracking to the estimated desired position and the attitude tracking to the command attitude are achieved. Moreover, an auxiliary system with proper parameters is implemented to guarantee the singularity-free command attitude extraction and to obviate the use of the unavailable desired information. The stability analysis and simulations effectively validate the achievement of the coordinated trajectory tracking of multiple VTOL UAVs with the proposed control approach.
\end{abstract}

\end{frontmatter}

\section{Introduction}
The past few decades have witnessed a rapid development in the formation control of unmanned aerial vehicles (UAVs). Replacing a single monolithic UAV with a formation of multiple micro ones can effectively improve efficiency without costly expense \citep{Giulietti2000}. More recently, the vertical takeoff and landing (VTOL) UAV, as a representative UAV, has received increasing interest, due to its capacities of hovering and low-speed/low-altitude flight \citep{Hua2013}. Additionally, the VTOL UAV is a canonical nonlinear system with under-actuation property  \citep{Zuo2010}, which raises a lot of technical problems for control theory research. Therefore, the formation control of multiple VTOL UAVs deserves intensive studies.

Generally, the study of formation control problem is categorized into leaderless and leader-follower formation control problem. The leaderless formation requires its members to simply reach a prescribed pattern \citep{Zhang2009}. For example, a distributed control algorithm is proposed in \cite{Abdessameud2009} such that the formation of VTOL UAVs with an identical velocity was achieved. The special case with communication delay was also studied \citep{Abdessameud2015} for the leaderless formation objective and corresponding control solutions were proposed.
Another formation protocol was developed in \cite{Dong2015} to realize a time-varying formation of VTOL UAVs without a leader, and the obtained theoretical results were verified with practical experiments.

Compared with the leaderless formation, the objective of the leader-follower formation is that followers reach an agreement with the desired information associated with a leader while forming the prescribed pattern \citep{Hong2008}. This may lead the formation to complete some intricate missions, where the leader is responsible for performing the desired trajectory of the formation and it is delivered via the communication network between the leader and followers. Although \cite{Yun2010,Mercado2013,Lee2012} proposed control approaches to achieve the
leader-follower formation of VTOL UAVs, the leader's desired information was required to be available to all the followers. In practice, due to limited information exchange and communication constraints, the leader's desired information is only accessible to a portion of the followers. To achieve the leader-follower formation under restricted
communication networks, distributed algorithms were implemented with a local information exchange mechanism \citep{Loria2016,Wen2016,Li2013,Yu2010,Qin2016,Yang2016}. Using backstepping and filtering strategies, a distributed control algorithm was developed in \citep{Ghommam2016} to realize the asymptotically stable leader-follower formation of VTOL UAVs.  A distributed formation and reconfiguration control approach is designed in \cite{Liao2017} to accomplish the leader-follower formation without inter-vehicle collisions. With feedback linearization technique, \cite{Mahnmood2015} proposed a leader-follower formation protocol for VTOL UAVs, which ensured their heading synchronization as well. \cite{Dong2016b} presented a distributed control strategy over a switched graph and derived necessary and sufficient conditions on the time-varying leader-follower formation of VTOL UAVs. However, the network graphs among the followers used in \cite{Ghommam2016,Liao2017,Mahnmood2015,Dong2016b} are undirected, which means that each pair of the followers interacts bidirectionally. This undirected graph condition is quite restrictive, which, due to communication constraints, can hardly be met in practical applications.

This paper proposes a coordinated trajectory tracking control approach for multiple VTOL UAVs with local information exchange, where the desired trajectory information is described by a leader. The network graph among the follower VTOL UAVs is assumed to be directed. This effectively relaxes the restrictive assumption that the graph is symmetric. By applying a novel distributed estimator, the leader's desired information is accurately estimated for each follower VTOL UAV. Based on the hierarchial framework, a command force and an applied torque are synthesized for each VTOL UAV such that the coordinated trajectory tracking is achieved for a group of VTOL UAVs.  Compared with the aforementioned work, the main contributions of this paper are three-fold. First, in contrast to the work in \cite{Abdessameud2009,Abdessameud2015,Dong2015}, where only a prescribed pattern is formed with a constant velocity, the leader-follower tracking of multiple VTOL UAVs is achieved by introducing a novel distributed estimator. Second, the coordinated tracking is achieved with weak connectivity, where the network graph among the followers is directed, rather than the limited undirected one used in \cite{Ghommam2016,Liao2017,Mahnmood2015,Dong2016b}. Third, instead of solely discussing the position loop \citep{Dong2015,Dong2016b}, a complete VTOL UAV system is studied based on a hierarchical framework, where an auxiliary system is proposed to ensure the non-singular command attitude extraction and to avoid the use of the unavailable desired information.

The remaining sections are arranged as follows. Section \ref{sec:2} describes the problem to be solved and provides some useful preliminaries. Section \ref{sec:3} states the main results in detail, including the distributed estimator design, the control approach project and the stability analysis. Section \ref{sec:4} performs some simulations to verify the theoretical results. And section \ref{sec:5} draws final conclusions.

\textbf{Notations.} 
$\mathbb{R}^{m\times n}$ denotes the $m\times n$ Euclidean space. Given a vector $x=[x_1,x_2,\cdots,x_n]^T$, define $\mathrm{sgn}(x)=[\mathrm{sgn}(x_1),\mathrm{sgn}(x_2),\cdots,\mathrm{sgn}(x_n)]^T$, and $\|x\|_1=\sum_{i=1}^n|x_i|$ and $\|x\|=\sqrt{x^Tx}$ are its $1$-norm and $2$-norm. Given a square matrix $A\!=\![a_{ij}]\!\in\!\mathbb{R}^{n\times n}$, define $\lambda_{\min}(A)$ and $\lambda_{\max}(A)$ as its minimum and maximum eigenvalues, and $\|A\|=\sqrt{\sum_{i=1}^n\sum_{j=1}^na_{ij}^2}$ is its F-norm. $I_n$ is an $n\times n$ identity matrix and  $\mathbf{1_n}$ is an $n$-dimensional vector with all entries being one.
Furthermore, given a vector $x=[x_1,x_2,x_3]^T$, superscript $\times$ represents a transformation from $x$ to a skew-symmetric matrix:
\begin{equation*}
x^\times=\left[\begin{array}{ccc}
0&-x_3&x_2\\x_3&0&-x_1\\-x_2&x_1&0
\end{array}\right].
\end{equation*}
\section{Background}
\label{sec:2}
\subsection{Problem statement}
Suppose that there are $n$ follower VTOL UAVs in a team, which are labeled by $\mathcal{V}=\{1,2,\dots,n\}$.
Each UAV is a six-dof (degree of freedom) rigid body and operates in two reference frames: inertia frame $\mathcal{I}=\{O_Ix_Iy_Iz_I\}$ which is attached to the earth and body frame $\mathcal{B}=\{O_Bx_By_Bz_B\}$ which is fixed to the fuselage. To establish the model of the UAVs, rotation matrix
$R_i\in\mathrm{SO}(3)\triangleq\{R\in\mathbb{R}^{3\times3}\mid\det R=1,R^TR=RR^T=I_3\}$ and unit quaternion $Q_i=[\sigma_i,q_i^T]^T\in\mathbb{Q}\triangleq\{Q\in\mathbb{R}\times\mathbb{R}^{3}\mid\sigma^2+q^Tq=1\}$ are applied to represent the attitude of each UAV. In terms of Euler formula \citep{Shuster1993}, an explicit relation between these two attitude representations is derived as
\begin{equation}
\label{Q_to_R}
R_i(Q_i)=(\sigma_i^2-q_i^Tq_i)I_3+2q_iq_i^T-2\sigma_iq_i^\times.
\end{equation}
Based on Euler-Newton formulae, the kinematics and dynamics of the $i$-th VTOL UAV are given by
\begin{align}
&\dot{p}_i=v_i,\label{pos_kin}\\
&\dot{v}_i=-g\hat{e}_3+\frac{T_i}{m_i}R_i(Q_i)\hat{e}_3,\label{pos_dyn}\\
&\dot{Q}_i=\frac{1}{2}G_i(Q_i)\omega_i,\label{att_kin}\\
&J_i\dot{\omega}_i=-\omega_i^\times J_i\omega_i+\tau_i,\label{att_dyn}
\end{align}
where $p_i=[p_i^x,p_i^y,p_i^z]^T$ and $v_i=[v_i^x,v_i^y,v_i^z]^T$ denote the position and velocity of the center of gravity of the UAV in frame $\mathcal{I}$, respectively, $m_i$ is the total mass, $g$ is the local gravitational acceleration, $\hat{e}_3\triangleq[0,0,1]^T$, $T_i$ denotes the applied thrust along $\hat{e}_3$, $Q_i=[\sigma_i,q_i^T]^T$ and $R_i(Q_i)$ are the unit quaternion and rotation matrix,
$G_i(Q_i)=[-q_i,\sigma_iI_3-q_i^\times]^T$, $\omega_i=[\omega_i^x,\omega_i^y,\omega_i^z]^T$ denotes the angular velocity of the UAV in frame $\mathcal{B}$, $J_i=\mathrm{diag}\{J_i^x,J_i^y,J_i^z\}$ is the inertial matrix with respect to frame $\mathcal{B}$, and $\tau_i$ denotes the applied torque in frame $\mathcal{B}$.

In addition to $n$ followers, there is a leader, labeled by $0$, to represent the global desired information including the desired position $p_r$ and its derivatives. The control objective is to design applied thrust $T_i$ and torque $\tau_i$ for each follower VTOL UAV described by \eqref{pos_kin}-\eqref{att_dyn} such that all the followers track the leader while maintaining a prescribed formation pattern. More specifically, given a desired position offset $\delta_i$, the objective is to guarantee that
\begin{equation}
\label{objective}
\lim_{t\rightarrow\infty}(p_i(t)-p_r(t))=\delta_i,~\lim_{t\rightarrow\infty}(v_i(t)-\dot{p}_r(t))=0,~\forall i\in\mathcal{V}.
\end{equation}
Due to communication constraints, the leader's desired information is only available to a subset of the followers and the followers only have access to their neighbors' information. To solve such a coordinated tracking problem via local information exchange, distributed algorithms are implemented. Moreover, it follows from \eqref{objective} that $\lim_{t\rightarrow\infty}(p_i(t)-p_j(t))\!=\!\delta_{ij}$, where $\delta_{ij}\!=\!\delta_i-\delta_j$, $\forall i,j\in\mathcal{V}$. This means that the followers form a pattern determined by $\delta_{i}$ while tracking the leader. Therefore, a proper position offset $\delta_{i}$ is essential such that the proposed algorithm ensures followers' convergence to a well-defined and unique formation.
\begin{assumption}
\label{assump:pr}
The desired position $p_r$ and its derivatives $\dot{p}_r$, $\ddot{p}_r$ and $p_r^{(3)}$ are bounded.
\end{assumption}

\subsection{Graph theory}
Communication topology among UAVs is described by a graph $\mathcal{G}_n\triangleq(\mathcal{V},\mathcal{E})$, which is composed of a node set $\mathcal{V}\triangleq\{1,2,\cdots,n\}$ and an edge set $\mathcal{E}\subseteq\mathcal{V}\times\mathcal{V}$. For a directed graph, $(i,j)\in\mathcal{E}$ means that the information of node $i$ is accessible to node $j$, but not conversely. All neighbours of node $i$ are included in set $\mathcal{N}_i=\{j\in\mathcal{V}\mid(j,i)\in\mathcal{E}\}$. A path from node $i$ to node $j$ is a sequence of edges.

For a follower graph $\mathcal{G}_n$, its adjacent matrix $\mathcal{D}=[d_{ij}]\in\mathbb{R}^{n\times n}$ is defined such that $d_{ij}>0$ if $(j,i)\!\in\!\mathcal{E}$ and $d_{ij}\!=\!0$ otherwise, and the associated nonsymmetric Laplacian matrix $\mathcal{L}\!=\![l_{ij}]\!\in\!\mathbb{R}^{n\times n}$ is defined such that $l_{ii}\!=\!\sum^n_{j=1,j\neq i}d_{ij}$ and $l_{ij}=-d_{ij}$ for $j\neq i$. For a leader-follower graph $\mathcal{G}_{n+1}\triangleq\{\bar{\mathcal{V}},\bar{\mathcal{E}}\}$ (leader is labeled as 0) with $\bar{\mathcal{V}}=\{0,1,\cdots,n\}$ and $\bar{\mathcal{E}}\subseteq\bar{\mathcal{V}}\times\bar{\mathcal{V}}$, we define $\bar{\mathcal{D}}\in\mathbb{R}^{(n+1)\times(n+1)}$ and $\bar{\mathcal{L}}\in\mathbb{R}^{(n+1)\times(n+1)}$ as its adjacent matrix and nonsymmetric Laplacian matrix. Specifically, $\bar{D}\triangleq\left[\begin{array}{cc}0&0_{1\times n}\\d_0&\mathcal{D}\end{array}\right]$, where $d_0=[d_{10},d_{20},\cdots,d_{n0}]^T$ and $d_{i0}>0$ if node $i$ is accessible to the leader and $d_{i0}=0$ otherwise; and $\bar{\mathcal{L}}\triangleq\left[\begin{array}{cc}0&0_{1\times n}\\-d_0&\mathcal{M}\end{array}\right]$, where $\mathcal{M}=[m_{ij}]\triangleq\mathcal{L}+\mathrm{diag}\{d_{10},d_{20},\cdots,d_{n0}\}$.
\begin{assumption}
\label{assump:graph}
The leader-follower graph $\mathcal{G}_{n+1}$ has a directed spanning tree with the leader being the root.
\end{assumption}
Some important properties associated with matrix $\mathcal{M}$ are given in Lemma \ref{lemma:graph} \citep{Qu2009}.
\begin{lemma}
\label{lemma:graph}
Under Assumption \ref{assump:graph}, $\mathcal{M}$ is a non-singular \textsc{M}-matrix with the properies that all its eigenvalues have positive real parts, and there exists a positive definite diagonal matrix $\Theta=\mathrm{diag}\{\theta_1,\theta_2,\cdots,\theta_n\}$ such that $\Xi\!=\!\mathcal{M}^T\Theta+\Theta\mathcal{M}$ is strictly diagonally dominant and positive definite, where
$[1/\theta_1,1/\theta_2,\cdots,1/\theta_n]^T\!=\!\mathcal{M}^{-1}\mathbf{1_n}$.
\end{lemma}

\subsection{Filippov solution and non-smooth analysis}
Consider the vector differential equation
\begin{equation}
\dot{x}=f(x,t),\label{sys}
\end{equation}
where $f:\mathbb{R}^n\times\mathbb{R}\rightarrow\mathbb{R}^n$ is measurable and essentially locally
bounded.

In what follows, the definitions of Filippov solution, generalized gradient and regular function are given according to  \cite{Paden1987,Shevitz1994,Clarke1983}.
\begin{definition}[\textbf{Filippov solution}]
A vector function $x(t)$ is called a solution of \eqref{sys} on $[t_0,t_1]$, if $x(t)$ is absolutely continuous on $[t_0,t_1]$ and for almost all $t\in[t_0,t_1]$, $\dot{x}\in\mathbb{K}[f](x,t)$, where
\begin{equation*}
\mathbb{K}[f](x,t)=\bigcap_{\rho>0}\bigcap_{\mu N=0}\overline{\mathrm{co}}f(B(x,\rho)-N,t),
\end{equation*}
$\bigcap_{\mu N=0}$ denotes the intersection over all sets $N$ of Lebesgue measure zero, $\overline{\mathrm{co}}(\cdot)$ denotes the vector convex closure, and $B(x,\rho)$ denotes
the open ball of radius $\rho$ centered at $x$.
\end{definition}

\begin{definition}[\textbf{Generalized gradient}]
\label{def:gradient}
For a locally Lipschitz function $V:\mathbb{R}^n\times\mathbb{R}\rightarrow\mathbb{R}$, its generalized
gradient at $(x,t)$ is defined as
\begin{equation*}
\partial V(x,t)\!=\!\overline{\mathrm{co}}\{\lim\triangledown V(x,t)\mid(x_i,t_i)\!\rightarrow\!(x,t),(x_i,t_i)\!\notin\!\Omega_V\},
\end{equation*}
where $\Omega_V$ is the set of measure zero where the gradient of $V$ is not defined. Furthermore, the generalized derivative of $V$ along system \eqref{sys} is defined as $\dot{\tilde{V}}\triangleq\bigcap_{\phi\in\partial V}\phi^T\left[\begin{array}{c}\mathbb{K}[f](x,t)\\1\end{array}\right]$.
\end{definition}

\begin{definition}[\textbf{Regular}]
$f(x,t):\mathbb{R}^n\times\mathbb{R}$ is called regular if\\
(1) for all $\nu\geq0$, the usual one-sided directional derivative $f'(x;\nu)$ exists;\\
(2) for all $\nu\geq0$, $f'(x;\nu)=f^o(x;\nu)$, where the generalized directional derivative $f^o(x;\nu)$ is defined as
\begin{equation*}
f^o(x;\nu)=\lim_{y\rightarrow x}\sup_{t\downarrow0}\frac{f(y+t\nu)-f(y)}{t}.
\end{equation*}
\end{definition}

The Lyapunov stability criterion for non-smooth systems is given in Lemma \ref{lemma:sys} \citep{Fischer2013}.
\begin{lemma}
\label{lemma:sys}
Let system \eqref{sys} be essentially locally bounded and $0\in\mathbb{K}[f](x,t)$ in a region $\mathbb{R}^n\times[0,\infty)$. Suppose that $f(0,t)$ is uniformly bounded for all $t\geq0$.  Let $V:\mathbb{R}^n\times [0,\infty)\rightarrow\mathbb{R}$ be locally Lipschitz in $t$ and regular such that $W_1(x)\leq V(t,x)\leq W_2(x)$ and $\dot{\tilde{V}}(x,t)\leq-W(x)$, where $W_1(x)$ and $W_2(x)$ are continuous positive definite functions, $W(x)$ is a continuous positive semi-definite function, and $\dot{\tilde{V}}(x,t)\leq-W(x)$ means that $\varphi\leq-W$, $\forall \varphi\in\dot{\tilde{V}}$. Then, all Filippov solutions of system \eqref{sys} are bounded and satisfy $\lim_{t\rightarrow\infty}W(x(t))=0$.
\end{lemma}

\section{Main results}
\label{sec:3}
Due to the under-actuated nature of the VTOL UAV, a hierarchical strategy is applied to solve the coordinated trajectory tracking problem of multiple VTOL UAV systems. First, a distributed estimator using local information interaction is designed for each follower UAV to estimate the leader's desired information. Then, the coordinated trajectory tracking problem of multiple VTOL UAVs is transformed into the asymptotic stability problem of each individual error system. Next, a command force and an applied torque are exploited for each UAV to asymptotically stabilize the position and attitude error systems, respectively. Finally, the stability of each error system is analyzed.

\subsection{Distributed estimator design}
Since the leader's desired information including the desired position $p_r$ and its derivatives is not available to all the followers, a distributed estimator is firstly designed for each VTOL UAV to estimate them.

For $i\in\mathcal{V}$, we define $\hat{p}_i$, $\hat{v}_i$ and $\hat{a}_i\!=\!k_\gamma\tanh(\gamma_i)$ as the estimations of $p_r$, $\dot{p}_r$ and $\ddot{p}_r$, respectively, where $\gamma_i$ is an auxiliary variable and parameter $k_\gamma\geq\sup_{t\geq0}\ddot{p}_r(t)$. As will be shown subsequently, the definition of $\hat{a}_i$ using the hyperbolic tangent function enables the control parameters to be chosen explicitly in case of singularity in the command attitude extraction. For $i\in\mathcal{V}$, a distributed estimator is proposed as follows:
\begin{subequations}
\label{estimator}
\begin{align}
\dot{\hat{p}}_i=&\hat{v}_i-k_p\sum_{j=0}^nd_{ij}(\hat{p}_i-\hat{p}_j), \\
\dot{\hat{v}}_i=&\hat{a}_i-k_v\sum_{j=0}^nd_{ij}(\hat{v}_i-\hat{v}_j), \\
\begin{split}
\ddot{\gamma}_i=&-l_a\dot{\gamma}_i+2\Gamma_i\dot{\gamma}_i\\
&-\!\frac{k_a}{k_\gamma}\bar{\Gamma}_i^{-1}
\left[\left(\sum_{j=0}^nd_{ij}(\hat{a}_i\!-\!\hat{a}_j)\!+\!\sum_{j=0}^nd_{ij}(\dot{\hat{a}}_i\!-\!\dot{\hat{a}}_j)\right)\right.\\
&\left.+\mathrm{sgn}\left(\sum_{j=0}^nd_{ij}(\hat{a}_i-\hat{a}_j)+\sum_{j=0}^nd_{ij}(\dot{\hat{a}}_i-\dot{\hat{a}}_j)\right)\right],
\end{split}
\label{estimator_c}
\end{align}
\end{subequations}
where  $\hat{p}_0=p_r$, $\hat{v}_0=\dot{p}_r$ and $\hat{a}_0=\ddot{p}_r$ are specified, $d_{ij}$ is the $(i,j)$-th entry of the adjacent matrix $\mathcal{D}$ associated with the follower graph $\mathcal{G}_{n}$, $k_p$, $k_v$, $k_a$ and $l_a$ are positive parameters, and $\Gamma_i=\mathrm{diag}\{\mu_i^x,\mu_i^y,\mu_i^z\}$ with $\mu_i^k=\tanh(\dot{\gamma}^k_i)\dot{\gamma}^k_i$ and
$\bar{\Gamma}_i=\mathrm{diag}\{\bar{\mu}_i^x,\bar{\mu}_i^y,\bar{\mu}_i^z\}$ with $\bar{\mu}_i^k=1-\tanh^2(\gamma^k_i)$ for $k=x,y,z$. Next, define the estimation errors $\bar{p}_i=\hat{p}_i-p_r$, $\bar{v}_i=\hat{v}_i-\dot{p}_r$ and $\bar{a}_i=\hat{a}_i-\ddot{p}_r$ for $i\in\mathcal{V}$. It then follows from \eqref{estimator} that their dynamics satisfy
\vspace{-20pt}
\begin{subequations}
\label{estimator_error}
\small{\begin{align}
\dot{\bar{p}}_i=&\bar{v}_i-k_p\sum_{j=1}^nm_{ij}\bar{p}_j,\\
\dot{\bar{v}}_i=&\bar{a}_i-k_v\sum_{j=1}^nm_{ij}\bar{v}_j,\\
\ddot{\bar{a}}_i=&k_\gamma\bar{\Gamma}_i\ddot{\gamma}_i-2k_\gamma\bar{\Gamma}_i\Gamma_i\dot{\gamma}_i-p_r^{(4)}\notag\\
=&-l_ak_\gamma\bar{\Gamma}_i\dot{\gamma}_i-k_a\left(\sum_{j=1}^nm_{ij}\bar{a}_j+\sum_{j=1}^nm_{ij}\dot{\bar{a}}_j\right)\notag\\
&-k_a\mathrm{sgn}\left(\sum_{j=1}^nm_{ij}\bar{a}_j+\sum_{j=1}^nm_{ij}\dot{\bar{a}}_j\right)-p_r^{(4)}\notag\\
\begin{split}
=&-l_a\dot{\bar{a}}_i-k_a\left(\sum_{j=1}^nm_{ij}\bar{a}_j+\sum_{j=1}^nm_{ij}\dot{\bar{a}}_j\right)\\
&-k_a\mathrm{sgn}\left(\sum_{j=1}^nm_{ij}\bar{a}_j+\sum_{j=1}^nm_{ij}\dot{\bar{a}}_j\right)+N_p,\\
\end{split}
\end{align}}
\end{subequations}
where $m_{ij}$ denotes the $(i,j)$-th entry of $\mathcal{M}$ defined in Section \ref{sec:2}, and $N_p=l_ap_r^{(3)}-p_r^{(4)}$ is bounded according to Assumption \ref{assump:pr}. Equivalently, the error dynamics \eqref{estimator_error} can be rewritten as
\begin{subequations}
\label{estimator_error_1}
\begin{align}
\dot{\bar{p}}=&\bar{v}-k_p(\mathcal{M}\otimes I_3)\bar{p},\\
\dot{\bar{v}}=&\bar{a}-k_v(\mathcal{M}\otimes I_3)\bar{v},\\
\begin{split}
\ddot{\bar{a}}=&-l_a\dot{\bar{a}}-k_a(\mathcal{M}\otimes I_3)(\bar{a}+\dot{\bar{a}})\\
&-k_a\mathrm{sgn}\left((\mathcal{M}\otimes I_3)(\bar{a}+\dot{\bar{a}})\right)+\mathbf{1_n}\otimes N_p,
\end{split}\label{estimator_error_c}
\end{align}
\end{subequations}
where $\bar{p}$, $\bar{v}$ and $\bar{a}$ are the column stack vectors of
 $\bar{p}_i$, $\bar{v}_i$ and $\bar{a}_i$, respectively, and operator $\otimes$ denotes the kronecker product. Moreover, define a sliding surface $s_i=l_a\sum_{j=1}^nm_{ij}\bar{a}_j+\sum_{j=1}^nm_{ij}\dot{\bar{a}}_j$ for $i\in\mathcal{V}$, and correspondingly, its column stack vector $s=(\mathcal{M}\otimes I_3)(l_a\bar{a}+\dot{\bar{a}})$. It follows from \eqref{estimator_error_c} that the dynamics of $s$ satisfies
\begin{equation}
\dot{s}=(\mathcal{M}\otimes I_3)\left(-k_as-k_a\mathrm{sgn}(s)+\mathbf{1_n}\otimes N_p\right).\label{s_dyn}
\end{equation}
Theorem \ref{theorem:1} indicates that the developed distributed estimator \eqref{estimator} with appropriate parameters enables the estimation errors $\bar{p}_i$, $\bar{v}_i$ and  $\bar{a}_i$ for each VTOL UAV to converge to zero asymptotically.
\begin{theorem}
\label{theorem:1}
Under Assumptions \ref{assump:pr} and \ref{assump:graph}, if the estimator parameters $k_p$, $k_v$, $l_a$ and $k_a$ are chosen  based on
\begin{align}
&k_pk_v>\frac{\|\Theta\|^2}{\lambda_{\min}(\Xi)^2},\label{k_pv}\\
&l_a>\frac{k_p\lambda_{\min}(\Xi)\|\Theta\|^2}{\lambda_{\min}(\mathcal{M}^T\Theta\mathcal{M})(k_pk_v\lambda_{\min}(\Xi)^2-\|\Theta\|^2)},\label{l_a}\\
&k_a>\frac{2\sqrt{n}\|\Theta\mathcal{M}\|\bar{N}_p}{\lambda_{\min}(\Xi)},\label{beta}
\end{align}
where $\bar{N}_p\!=\!\sup_{t\geq0}\|N_p(t)\|$, and $\Theta$ and $\Xi$ are given in Lemma \ref{lemma:graph}, the distributed estimator \eqref{estimator} ensures that $\lim_{t\rightarrow\infty}\bar{p}_i(t)\!=\!0$, $\lim_{t\rightarrow\infty}\bar{v}_i(t)\!=\!0$ and $\lim_{t\rightarrow\infty}\bar{a}_i(t)\!=\!0$, $\forall i\!\in\!\mathcal{V}$.
\end{theorem}
\begin{proof}
The proof is divided into two parts: first, the sliding surface $s_i$, $i\in\mathcal{V}$ is proven in Proposition \ref{pro:s} to converge to zero asymptotically; then, the final result is shown in Proposition \ref{pro:pva}.
\end{proof}

\begin{proposition}
\label{pro:s}
Under Assumptions \ref{assump:pr} and \ref{assump:graph}, if the estimator parameter $k_a$ satisfies \eqref{beta}, the distributed estimator \eqref{estimator} guarantees that $\lim_{t\rightarrow\infty}s_i(t)\!=\!0$, $\forall i\in\mathcal{V}$.
\end{proposition}
\begin{proof}
Obviously, system \eqref{s_dyn} is non-smooth; thereafter, the solution of \eqref{s_dyn} is studied in the sense of Filippov and the non-smooth framework is applied. The stability of system \eqref{s_dyn} is to be proven based on Lemma \ref{lemma:sys}.

We first propose a Lyapunov function $L^s\!=\!\sum_{i=1}^n\theta_i\|s_i\|_1\!+s^T(\Theta\otimes I_3)s$,
where $\theta_i$ is the $i$-th diagonal entry of $\Theta$. Note that $L^s$ is non-smooth but regular \citep{Paden1987}.
It can be derived that $L^s$ is bounded by
\begin{equation*}
W_1(s)\leq L^s\leq W_2(s),
\end{equation*}
where $W_1(s)=\lambda_{\min}(\Theta)\|s\|^2$ and $W_2(s)=\sqrt{n}\lambda_{\max}(\Theta)\\\cdot(\|s\|_1+\|s\|^2)$. In terms of Lemma \ref{lemma:sys}, the stable result can be deduced if only the generalized derivative of $L^s$ satisfies $\dot{\tilde{L}}^s\leq W(s)$, where $W(s)$ is a continuous positive semi-definite function.

According to Definition \ref{def:gradient}, the generalized derivative of $L^s$ along \eqref{s_dyn} satisfies
\begin{align}
\dot{\tilde{L}}^s=\bigcap_{\mathclap{\phi\in\partial L^s}}&(\phi+s)^T(\Theta\otimes I_3)\mathbb{K}[(\mathcal{M}\otimes I_3)(-k_as-k_a\mathrm{sgn}(s)\notag\\
&+\mathbf{1_n}\otimes N_p)],\notag\\
=\!\bigcap_{\mathclap{\phi\in\partial L^s}}&(\phi+s)^T(\Theta\mathcal{M}\!\otimes I_3)(-k_as\!-k_a\partial\|s\|_1+\mathbf{1_n}\!\otimes N_p),
\end{align}
where $\partial|s_i^j|=\begin{cases}\{1\}&s_i^j\in\mathbb{R}^+\\\{-1\}&s_i^j\in\mathbb{R}^-\\ [-1,1]&s_i^j=0\end{cases},\forall i\in\mathcal{V}$, $j=x,y,z$, and the calculation of $\mathbb{K}$ is applied using the same argument given in \cite{Paden1987}.

If $\dot{\tilde{L}}^s\neq\emptyset$, suppose $\varphi\notin\dot{\tilde{L}}^s$, then we know that
$\forall \phi\in\partial\|s\|_1$, $\varphi=(\phi+s)^T(\Theta\mathcal{M}\otimes I_3)(-k_as-k_a\delta+\mathbf{1_n}\otimes N_p)$ for some $\delta\in\partial\|s\|_1$. Choose
$\phi=\arg\min_{\delta\in\partial\|s\|_1}(\delta+s)^T(\frac{\Theta\mathcal{M}+\mathcal{M}^T\Theta}{2}\otimes I_3)(\delta+s)$. According to \cite{Paden1987}, for all $\delta\in\partial\|s\|_1$, we have that
\begin{align*}
&(\phi+s)^T(\Theta\mathcal{M}\otimes I_3)(\delta+s)\notag\\
\geq&(\phi+s)^T\left(\frac{\Theta\mathcal{M}+\mathcal{M}^T\Theta}{2}\otimes I_3\right)(\phi+s)\notag\\
=&\frac{1}{2}(\phi+s)^T(\Xi\otimes I_3)(\phi+s).
\end{align*}
It then follows that $\dot{\tilde{L}}^s$ further satisfies
\begin{align}
\dot{\tilde{L}}^s\leq&-\frac{k_a}{2}(\phi+s)^T(\Xi\otimes\! I_3)(\phi+s)+\!(\phi+s)^T(\mathbf{1_n}\otimes\! N_p)\notag\\
\leq&-\frac{k_a\lambda_{\min}(\Xi)}{2}\|\phi+s\|^2+\sqrt{n}\|\phi+s\|\|\Theta\mathcal{M}\|\bar{N}_p\notag\\
=&-\left(\frac{k_a\lambda_{\min}(\Xi)}{2}\|\phi+s\|-\sqrt{n}\|\Theta\mathcal{M}\|\bar{N}_p\right)\|\phi+s\|.
\end{align}
Note that if $\phi=0$, then $s=0$ and $\|\phi+s\|=0$, and if $\phi\neq0$, then $\|\phi+s\|\geq1$. Hence, if the estimator parameter $k_a$ satisfies \eqref{beta}, there exists a constant $\bar{k}_a$ satisfying
$\frac{2\sqrt{n}\|\Theta \mathcal{M}\|\bar{N}_p}{\lambda_{\min}(\Xi)}\leq\bar{k}_a<k_a$
such that $-(\frac{\bar{k}_a\lambda_{\min}(\Xi)}{2}\|\phi+s\|-\sqrt{n}\|\Theta\mathcal{M}\|\bar{N}_p)\|\phi+s\|\leq0$. Therefore, it follows that
\begin{equation}
\dot{\tilde{L}}^s\leq-\frac{(k_a-\bar{k}_a)\lambda_{\min}(\Xi)}{2}\|\phi+s\|^2.
\end{equation}
In addition, each entry of $s$ has the same sign as its counterpart in $\mathrm{sgn}(s)$, and thus, it follows that  $\|\phi+s\|\geq\|s\|$. We finally have that
\begin{equation}
\dot{\tilde{L}}^s\leq-\frac{(k_a-\bar{k}_a)\lambda_{\min}(\Xi)}{2}\|s\|^2=-W(s).
\end{equation}
Since $W(s)\geq0$ has been ensured, it follows that $\int_{0}^tW(s(\tau))\mathrm{d}\tau$ is bounded for $t\geq0$, which implies that $s$ is bounded. Hence, it follows from \eqref{s_dyn} that $\dot{s}$ is bounded, which implies that $s$ is uniformly continuous in $t$.  This means that $W(s)$ is uniformly continuous in $t$. Based on Lemma \ref{lemma:sys}, we can conclude that $\lim_{t\rightarrow\infty}W(s(t))=0$, which further implies that $\lim_{t\rightarrow\infty}s(t)=0$.
\end{proof}

\begin{proposition}
\label{pro:pva}
Under Assumptions \ref{assump:pr} and \ref{assump:graph}, if the estimator parameters $k_p$, $k_v$ and $l_a$ satisfy \eqref{k_pv} and \eqref{l_a}, then $\lim_{t\rightarrow\infty}s_i(t)\!=\!0$ is sufficient to ensure that $\lim_{t\rightarrow\infty}\bar{p}_i(t)\!=\!0$, $\lim_{t\rightarrow\infty}\bar{v}_i(t)=0$ and $\lim_{t\rightarrow\infty}\bar{a}_i(t)=0$, $\forall i\in\mathcal{V}$.
\end{proposition}
\begin{proof}
Consider the definition of the sliding surface $s$, then the dynamics of the estimator error $\bar{a}$ satisfies
\begin{equation}
(\mathcal{M}\otimes I_3)\dot{\bar{a}}=-l_a(\mathcal{M}\otimes I_3)\bar{a}+s.\label{error_a}
\end{equation}
Define $z\!=\![\bar{p}^T,\bar{v}^T,\bar{a}^T]^T$ and assign a Lyapunov function
\begin{align*}
L^e=\bar{p}^T(\Theta\otimes I_3)\bar{p}+\bar{v}^T(\Theta\!\otimes\! I_3)\bar{v}+\frac{1}{2}\bar{a}^T(\mathcal{M}^T\Theta\mathcal{M}\otimes I_3)\bar{a}.
\end{align*}
It is bounded by
\begin{equation*}
\lambda_1\|z\|^2\leq L^e\leq\lambda_2\|z\|^2,
\end{equation*}
where $\lambda_1\!=\!\min\{\lambda_{\min}(\Theta),\frac{1}{2}\lambda_{\min}(\mathcal{M}^T\Theta\mathcal{M})\}$ and $\lambda_2=\max\{\lambda_{\max}(\Theta),\frac{1}{2}\lambda_{\max}(\mathcal{M}^T\Theta\mathcal{M})\}$.
The derivative of $L^e$ along \eqref{estimator_error_1} and \eqref{error_a} satisfies
\begin{align}
\label{dL_e}
\dot{L}^e=&-\bar{p}^T((\mathcal{M}^T\Theta+\Theta\mathcal{M})\otimes I_3)\bar{p}+2\bar{p}^T(\Theta\otimes I_3)\bar{v}\notag\\
&-\bar{v}^T((\mathcal{M}^T\Theta+\Theta\mathcal{M})\otimes I_3)\bar{v}+2\bar{v}^T(\Theta\otimes I_3)\bar{a}\notag\\
&+\bar{a}^T(\mathcal{M}^T\Theta\otimes I_3)(-l_a(\mathcal{M}\otimes I_3)\bar{a}+s)\notag\\
=&-k_p\bar{p}^T(\Xi\otimes I_3)\bar{p}+2\bar{p}^T(\Theta\otimes I_3)\bar{v}-k_v\bar{v}^T(\Xi\otimes I_3)\bar{v}\notag\\
&+2\bar{v}^T(\Theta\otimes I_3)\bar{a}-l_a\bar{a}^T(\mathcal{M}^T\Theta\mathcal{M}\otimes I_3)\bar{a}\notag\\
&+\bar{a}^T(\mathcal{M}^T\Theta\otimes I_3)s\notag\\
\leq&-k_p\lambda_{\min}(\Xi)\|\bar{p}\|^2+2\|\Theta\|\|\bar{p}\|\|\bar{v}\|-k_v\lambda_{\min}(\Xi)\|\bar{v}\|^2\notag\\
&+2\|\Theta\|\|\bar{v}\|\|\bar{a}\|-l_a\lambda_{\min}(\mathcal{M}^T\Theta\mathcal{M})\|\bar{a}\|^2\notag\\
&+\|\Theta\mathcal{M}\|\|\bar{a}\|\|s\|\notag\\
\leq&-z_1^T\Omega z_1+\|\Theta\mathcal{M}\|\|a\|\|s\|,
\end{align}
where $z_1=[\|\bar{p}\|,\|\bar{v}\|,\|\bar{a}\|]^T$ and
\begin{equation*}
\Omega=\left[\begin{array}{ccc}
k_p\lambda_{\min}(\Xi)&-\|\Theta\|&0\\
-\|\Theta\|&k_v\lambda_{\min}(\Xi)&-\|\Theta\|\\
0&-\|\Theta\|&-l_a\lambda_{\min}(\mathcal{M}^T\Theta\mathcal{M})\\
\end{array}\right].
\end{equation*}
If the estimator parameters $k_p$, $k_v$ and $l_a$ are chosen based on \eqref{k_pv} and \eqref{l_a}, then $\Omega$ is positive definite. In this case, $\dot{L}^e$ further satisfies
\begin{align}
\label{dL_e_1}
\dot{L}^e\leq&-\lambda_{\min}(\Omega)\|z\|^2+\|\Theta\mathcal{M}\|\|z\|\|s\|\notag\\
\leq&-\vartheta_1 L^e+\vartheta_2\|s\|\sqrt{L^e},
\end{align}
where $\vartheta_1=\frac{\lambda_{\min}(\Omega)}{\lambda_2}$, $\vartheta_2=\frac{\|\Theta\mathcal{M}\|}{\sqrt{\lambda_1}}$, and $\|z\|=\|z_1\|$ has been applied. Next, take $V^e=\sqrt{2L^e}$. When $L^e\neq0$, it follows from \eqref{dL_e_1} that the derivative of $V^e$ satisfies
\begin{equation}
\dot{V}^e\leq-\vartheta_1V^e+\vartheta_2\|s\|.\label{V_e}
\end{equation}
When $V^e=0$, it can be shown that $D^+V^e\leq\vartheta_2\|s\|$. Thus, it follows that $D^+V^e$ satisfies \eqref{V_e} all the time \citep{Khalil2002}. For system $\dot{y}=-\vartheta_1y+\vartheta_2\|s\|$ with respect to $y\in[0,\infty)$, it can be proven in terms of input-to-state stability theory \citep{Khalil2002} that $\lim_{t\rightarrow\infty}y(t)=0$ given the fact that $\lim_{t\rightarrow\infty}s(t)=0$. According to Comparison Principal \citep{Khalil2002}, it follows that $\lim_{t\rightarrow\infty}V^e(t)=0$, which further implies that $\lim_{t\rightarrow\infty}\bar{p}(t)=0$, $\lim_{t\rightarrow\infty}\bar{v}(t)=0$ and $\lim_{t\rightarrow\infty}\bar{a}(t)=0$.
\end{proof}

\begin{remark}
According to \eqref{estimator},  singularity may occur in the distributed estimator when some diagonal entry of $\bar{\Gamma}_i$ equals to zero, and this corresponds to the case where some entry of the auxiliary variable $\gamma_i$ tends to infinity. Theorem \ref{theorem:1} has shown that, with a bounded initial value, the estimation error $\bar{a}_i$ for each UAV is bounded all the time. This implies that $\bar{\Gamma}_i$ is always positive definite. Consequently, no singularity is introduced in the developed distributed estimator \eqref{estimator}.
\end{remark}

\subsection{Problem transformation}
\label{sec:3.2}
Since the leader's desired information has been estimated via the distributed estimator \eqref{estimator} for each follower VTOL UAV, the remaining problem is to transform the coordinated trajectory tracking problem into the simultaneous tracking problem for the decoupled VTOL UAV group. This is verified as follows.

Define the position error $p_i^e=p_i-p_r-\delta_i$ and the velocity error $v_i^e=v_i-\dot{p}_r$ for $i\in \mathcal{V}$. Using the estimations $\hat{p}_i$ and $\hat{v}_i$ obtained from the distributed estimator \eqref{estimator}, we rewrite $p_i^e$ and $v_i^e$ as $p_i^e=p_i-\hat{p}_i-\delta_i+\bar{p}_i$ and $v_i^e=v_i-\hat{v}_i+\bar{v}_i$.
Since Theorem \ref{theorem:1} has shown that $\lim_{t\rightarrow\infty} \bar{p}_i(t)=0$ and $\lim_{t\rightarrow\infty} \bar{v}_i(t)=0$, $\forall i\in\mathcal{V}$, the coordinated tracking control objective \eqref{objective} can be transformed into the following simultaneous tracking objective:
\begin{equation*}
\lim_{t\rightarrow\infty}(p_i(t)-\hat{p}_i(t))=\delta_i, \lim_{t\rightarrow\infty}(v_i(t)-\hat{v}_i(t))=0, \forall i\in\mathcal{V}.
\end{equation*}
We next define $\bar{p}_i^e=p_i-\hat{p}_i-\delta_i$ and $\bar{v}_i^e=v_i-\hat{v}_i$ for $i\in\mathcal{V}$. It follows from \eqref{pos_kin}, \eqref{pos_dyn} and \eqref{estimator} that their dynamics satisfy
\begin{subequations}
\label{pos_err_sys}
\begin{align}
\dot{\bar{p}}_i^e=&\bar{v}_i^e+\sum^n_{j=1}m_{ij}\bar{p}_j,\\
\dot{\bar{v}}_i^e\!=&u_i-g\hat{e}_3-\hat{a}_i+\sum^n_{j=1}m_{ij}\bar{v}_j+\frac{T_i}{m_i}(R_i(Q_i)
\!\!-R_i(Q_i^c))\hat{e}_3,
\end{align}
\end{subequations}
where $u_i=\frac{T_i}{m_i}R_i(Q_i^c)\hat{e}_3$ is the command force with $Q_i^c$ being the command unit quaternion. Moreover, once the command force $u_i$ can be determined, in view of $\|R_i(Q_i^c)\hat{e}_3\|=1$, the applied thrust $T_i$ is derived as
\begin{equation}
\label{thrust}
T_i=m_i\|u_i\|,~~\forall i\in\mathcal{V}.
\end{equation}
Now that the control strategy is based on a hierarchical framework, the command unit quaternion $Q_i^c$, as the attitude tracking objective for each VTOL UAV, should be extracted from the command force $u_i$. Based on minimal rotation principle, a viable extraction algorithm is proposed in Lemma \ref{lemma:Qc} \citep{Abdessameud2009}.
\begin{lemma}
\label{lemma:Qc}
For $i\in\mathcal{V}$, if the command force $u_i=[u_i^x,u_i^y,u_i^z]^T$ satisfies the non-singular condition:
\begin{equation}
\label{u_con}
u_i\notin\mathcal{U}\triangleq\{u\in\mathbb{R}^3\mid u=[0,0,u^z]^T,u^z\leq0\},
\end{equation}
 the command unit quaternion $Q_i^c=[\sigma_i^c,(q_i^c)^T]^T$ is extracted as
\begin{align}
\sigma_i^c=\sqrt{\frac{1}{2}+\frac{g-u^z_i}{2\|u_i\|}},~
q_i^c=\frac{1}{2\|u_i\|\sigma_i^c}\left[\begin{array}{c}u_i^y\\-u_i^x\\0\end{array}\right].
\label{Q_c}
\end{align}
\end{lemma}
Next, define the attitude error $Q^e_i=[\sigma_i^e,(q_i^e)^T]^T=(Q^c_i)^{-1}\odot Q_i$ for $i\in\mathcal{V}$, where operator $\odot$ is the unit quaternion product. According to \cite{Zou2016}, $Q_i^e=[\pm 1,0,0,0]^T$ corresponds to the extract attitude tracking. The dynamics of $Q_i^e$ satisfies
\begin{equation}
\label{att_err_kin}
\dot{Q}_i^e=\frac{1}{2}G_i(Q_i^e)\omega_i^e,
\end{equation}
where $\omega_i^e=\omega_i-R_i(Q_i^e)^T\omega_i^c$ is the angular velocity error with $\omega_i^c$ being the command angular velocity. Please refer to \cite{Zou2016} for the derivations of $\omega_i^c$ and its derivative $\dot{\omega}_i^c$. In addition, it follows from \eqref{att_dyn} and $\dot{R}_i(Q_i^e)=R_i(Q_i^e)(\omega_i^e)^\times$ that the dynamics of $\omega_i^e$ satisfies
\begin{align}
\label{att_err_dyn}
J_i\dot{\omega}_i^e=&-\omega_i^\times J_i\omega_i+\tau_i+J_i[(\omega_i^e)^\times R_i(Q_i^e)^T\omega_i^c\notag\\
&-R_i(Q_i^e)^T\dot{\omega}_i^c].
\end{align}

Based on the above discussions, by introducing the distributed estimator \eqref{estimator}, the coordinated trajectory tracking problem for multiple VTOL UAV systems \eqref{pos_kin}-\eqref{att_dyn} can be transformed into the  simultaneous asymptotic stability problem for each error system \eqref{pos_err_sys}, \eqref{att_err_kin} and
\eqref{att_err_dyn}. Lemma \ref{lemma:trans} summarizes this point.
\begin{lemma}
\label{lemma:trans}
Consider the i-th error system \eqref{pos_err_sys}, \eqref{att_err_kin} and \eqref{att_err_dyn}.
If a command force $u_i$ and an applied torque $\tau_i$ can be developed such that $\lim_{t\rightarrow\infty}\bar{p}_i^e(t)=0$, $\lim_{t\rightarrow\infty}\bar{v}^e_i(t)=0$, $\lim_{t\rightarrow\infty}q_i^e(t)=0$ and $\lim_{t\rightarrow\infty}\omega_i^e(t)=0$, the coordinated
trajectory tracking of multiple VTOL UAV systems \eqref{pos_kin}-\eqref{att_dyn} is achieved in the sense of \eqref{objective}.
\end{lemma}

\subsection{Command force development}
\label{sec:3.3}

In this subsection, a command force $u_i$ for each VTOL UAV will be synthesized. The main difficulties here are that the command force $u_i$ should comply with the non-singular condition \eqref{u_con} and the desired position $p_r$ and its derivatives are not available in the command force $u_i$ and the subsequent applied torque $\tau_i$ due to limited communication.

To address the above dilemmas, we introduce the virtual position error $\tilde{p}_i^e=\bar{p}_i^e-\eta_i$ and the virtual velocity error $\tilde{v}_i^e=\bar{v}_i^e-\dot{\eta}_i$ for $i\in\mathcal{V}$, where $\eta_i$ is an auxiliary variable. It follows from \eqref{pos_err_sys} that the dynamics of $\tilde{p}_i^e$ and $\tilde{v}_i^e$ satisfy
\begin{subequations}
\label{pos_err_sys_1}
\begin{align}
\dot{\tilde{p}}_i^e=~&\tilde{v}_i^e+\hbar_1,\\
\dot{\tilde{v}}_i^e=~&u_i-\ddot{\eta}_i-g\hat{e}_3-\hat{a}_i+\hbar_2,
\end{align}
\end{subequations}
where $\hbar_1=\sum^n_{j=1}m_{ij}\bar{p}_j$ and $\hbar_2=\sum^n_{j=1}m_{ij}\bar{v}_j+\frac{T_i}{m_i}(R_i(Q_i)-R_i(Q_i^c))\hat{e}_3$.
\begin{lemma}
\label{lemma:trans_2}
Consider the $i$-th virtual position error system \eqref{pos_err_sys_1}. If a command force $u_i$ can be synthesized such
that $\lim_{t\rightarrow\infty}\tilde{p}_i^e(t)=0$, $\lim_{t\rightarrow\infty}\tilde{v}_i^e(t)=0$, $\lim_{t\rightarrow\infty}\eta_i(t)=0$ and $\lim_{t\rightarrow\infty}\dot{\eta}_i(t)=0$,
then $\lim_{t\rightarrow\infty}\bar{p}_i^e(t)=0$ and $\lim_{t\rightarrow\infty}\bar{v}_i^e(t)=0$ are achieved.
\end{lemma}

To guarantee the condition in Lemma \ref{lemma:trans_2}, for $i\in\mathcal{V}$, we propose the following command force:
\begin{equation}
\label{command_force}
u_i=g\hat{e}_3+\hat{a}_i-k_\eta(\tanh(\eta_i+\dot{\eta}_i)+\tanh(\dot{\eta}_i)),
\end{equation}
and introduce a dynamic system with respect to the auxiliary variable $\eta_i$:
\begin{align}
\label{aux_sys}
\ddot{\eta}_i=-k_\eta(\tanh(\eta_i+\dot{\eta}_i)+\tanh(\dot{\eta}_i))+l_p\tilde{p}_i^e+l_v\tilde{v}_i^e,
\end{align}
where $k_\eta$, $l_p$ and $l_v$ are positive control parameters. Substituting \eqref{command_force} and \eqref{aux_sys} into \eqref{pos_err_sys_1} yields
\begin{subequations}
\label{pos_err_sys_2}
\begin{align}
\dot{\tilde{p}}_i^e=~&\tilde{v}_i^e+\hbar_1,\\
\dot{\tilde{v}}^e=~&-l_p\tilde{p}_i^e-l_v\tilde{v}_i^e+\hbar_2.
\end{align}
\end{subequations}
A proper control parameter $k_\eta$ should be chosen such that the non-singular condition \eqref{u_con} is met. Specifically,
\begin{equation}
\label{par_con}
k_\eta<\frac{1}{2}(g-k_\gamma).
\end{equation}
In such a case, the third row of the command force $u_i$ satisfies
\begin{align*}
u_i^z=&g+\hat{a}_i^z-k_\eta(\tanh(\eta_i^z+\dot{\eta}_i^z)+\tanh(\dot{\eta}_i^z))\notag\\
\geq&g-k_\gamma-2k_\eta>0,
\end{align*}
where $\hat{a}_i^z=k_\gamma\tanh(\gamma_i^z)$ and the property that $|\tanh(\cdot)|<1$ have been applied.  To this end, $k_\eta$ satisfying \eqref{par_con} is sufficient to guarantee that the developed command force $u_i$ in \eqref{command_force} for each UAV strictly satisfies the non-singular condition \eqref{u_con}.

\begin{remark}
By defining $\hat{a}_i=k_\gamma\tanh(\gamma_i)$ in the distributed estimator \eqref{estimator} and introducing the auxiliary dynamics \eqref{aux_sys}, the developed command force $u_i$ for $i\in\mathcal{V}$ is equipped with a saturation property. Based on this property, the choice of the control parameter $k_\eta$  is independent on any estimator state.
\end{remark}

\begin{remark}
It follows from \eqref{thrust}, \eqref{command_force} and $\|\hat{a}_i\|\leq\sqrt{3}k_\gamma$ that the resulting applied thrust $T_i$  is bounded by
\begin{equation}
\label{T_bound}
T_i\leq m_i(g+2\sqrt{3}k_\eta+\sqrt{3}k_\gamma),~~\forall i\in\mathcal{V},
\end{equation}
which means that each $T_i$ is upper bounded by a constant associated with the individual mass $m_i$ and the specified   parameters $k_\eta$ and $k_\gamma$.
\end{remark}

\subsection{Applied torque development}
Define a sliding surface $r_i=l_qq_i^e+\omega_i^e$ for $i\in\mathcal{V}$,
where $l_q>0$. From \eqref{att_err_kin} and \eqref{att_err_dyn}, the dynamics of $r_i$ satisfies
\begin{align}
\label{ds_q}
\begin{split}
J_i\dot{r}_i=&\frac{l_q}{2}J_i(\sigma^e_iI_3+(q_i^e)^\times)\omega_i^e-\omega_i^\times J_i\omega_i+\tau_i\\
&+J_i[(\omega_i^e)^\times R_i(Q_i^e)^T\omega_i^c-R_i(Q_i^e)^T\dot{\omega}_i^c].
\end{split}
\end{align}
We propose an applied torque $\tau_i$ for each UAV as follows:
\begin{align}
\label{torque}
\begin{split}
\tau_i=&-k_qr_i-\frac{l_q}{2}J_i(\sigma^e_iI_3+(q_i^e)^\times)\omega_i^e+\omega_i^\times J_i\omega_i\\
&-J_i[(\omega_i^e)^\times R_i(Q_i^e)^T\omega_i^c-R_i(Q_i^e)^T\dot{\omega}_i^c],
\end{split}
\end{align}
where $k_q>0$. Substituting \eqref{torque} into \eqref{ds_q} yields
\begin{equation}
\label{ds_q1}
J_i\dot{r}_i=-k_qr_i.
\end{equation}
It follows from \eqref{torque} that, the command angular
velocity $\omega_i^c$ and its derivative $\dot{\omega}_i^c$ are necessary to determine each
applied torque $\tau_i$. According to \cite{Zou2016}, $\omega_i^c$ and $\dot{\omega}_i^c$ are the functions
of the derivatives of the command force $u_i$. Their expressions are presented as follows:
\begin{align*}
\begin{split}
\dot{u}_i=&k_\gamma\bar{\Gamma}_i\dot{\gamma}_i-k_\eta [D_i\dot{\eta}_i+(D_i+S_i)\ddot{\eta}_i],\\
\ddot{u}_i=&-2k_\gamma\bar{\Gamma}_i\Gamma_i\dot{\gamma}_i+k_\gamma\bar{\Gamma}_i\ddot{\gamma}_i-k_\eta[\bar{D}_i\dot{\eta}_i+(\bar{D}_i+{D}_i+\bar{S}_i)\ddot{\eta}_i\notag\\
&+(D_i+S_i)\eta^{(3)}_i],
\end{split}
\end{align*}
where $\Gamma_i$ and $\bar{\Gamma}_i$ have been specified below \eqref{estimator},
$D_i=\mathrm{diag}\{\epsilon_i^x,\epsilon_i^y,\epsilon_i^z\}$ with $\epsilon_i^j=1-\tanh^2(\eta_i^j+\dot{\eta}_i^j)$, $S_i=\mathrm{diag}\{\nu_i^x,\nu_i^y,\nu_i^z\}$ with $\nu_i^j=1-\tanh^2(\dot{\eta}_i^j)$, $\bar{D}_i=\{\bar{\epsilon}_i^x,\bar{\epsilon}_i^y,\bar{\epsilon}_i^z\}$ with $\bar{\epsilon}_i^j=-2\tanh(\eta_i^j+\dot{\eta}_i^j)(1-\tanh^2(\eta_i^j+\dot{\eta}_i^j))(\eta_i^j+\dot{\eta}_i^j)$ and
$\bar{S}_i=\{\bar{\nu}_i^x,\bar{\nu}_i^y,\bar{\nu}_i^z\}$ with
$\bar{\nu}_i^j=-2\tanh(\dot{\eta}_i^j)(1-\tanh^2(\dot{\eta}_i^j))\dot{\eta}_i^j$, for $j=x,y,z$, and
\begin{align*}
\eta_i^{(3)}=&-k_\eta[D_i\dot{\eta}_i+(D_i+S_i)\ddot{\eta_i}]+l_p\dot{\tilde{p}}^e_i+l_v\dot{\tilde{v}}^e_i.
\end{align*}
From the above derivations, it it trivial to see that the desired information is not used in the developed applied torque $\tau_i$ for the UAV without accessibility to the leader.

\subsection{Stability analysis}
Theorem \ref{theorem:2} summarizes the final stability result associated with the coordinated trajectory tracking of multiple VTOL UAV systems \eqref{pos_kin}-\eqref{att_dyn} controlled by the developed applied thrust and torque.

\begin{theorem}
\label{theorem:2}
Consider $n$ follower VTOL UAV systems \eqref{pos_kin}-\eqref{att_dyn} with Assumptions \ref{assump:pr} and \ref{assump:graph}. The synthesized command force $u_i$ in \eqref{command_force} and applied
torque $\tau_i$ in \eqref{torque} guarantee the coordinated trajectory tracking of multiple
VTOL UAVs in the sense of \eqref{objective}.
\end{theorem}
\begin{proof}
Theorem \ref{theorem:1} has shown that, for $i\in\mathcal{V}$, the distributed estimator developed in \eqref{estimator} enables the estimation errors $\bar{p}_i$ and $\bar{v}_i$  to converge to zero asymptotically.
Based on this, it follows from Lemmas \ref{lemma:trans} and \ref{lemma:trans_2} that the coordinated trajectory tracking objective is achieved, if the following results are guaranteed by the synthesized command force $u_i$ and applied torque $\tau_i$:\\
~~~~ Th2.i) $\lim_{t\rightarrow\infty}q_i^e(t)=0$ and $\lim_{t\rightarrow\infty}\omega_i^e(t)=0$, $\forall i\in\mathcal{V}$,\\
~~~~ Th2.ii) $\lim_{t\rightarrow\infty}\tilde{p}_i^e(t)=0$ and $\lim_{t\rightarrow\infty}\tilde{v}_i^e(t)=0$, $\forall i\in\mathcal{V}$,\\
~~~~ Th2.iii) $\lim_{t\rightarrow\infty}\eta_i(t)=0$ and $\lim_{t\rightarrow\infty}\dot{\eta}_i(t)=0$, $\forall i\in\mathcal{V}$.\\
They will be proven in Propositions \ref{pro:1}-\ref{pro:3}, respectively.

\begin{proposition}
\label{pro:1}
Consider the attitude error system \eqref{att_err_kin} and \eqref{att_err_dyn}.  The developed applied torque $\tau_i$ in \eqref{torque} guarantees that $\lim_{t\rightarrow\infty}q_i^e(t)=0$ and $\lim_{t\rightarrow\infty}\omega_i^e(t)=0$, $\forall i\in\mathcal{V}$.
\end{proposition}
\begin{proof}
It follows from \eqref{ds_q1} that, for $i\in\mathcal{V}$, the developed applied torque $\tau_i$ enables the sliding surface $r_i$ to converge to zero asymptotically. Then, assign a non-negative function $y_i=\frac{1}{2}[\|q_i^e\|^2+(1-\sigma_i^e)^2]=1-\sigma_i^e\leq2$ for $i\in\mathcal{V}$. With the definition of $r_i$, the derivative of $y_i$ along \eqref{att_err_kin} satisfies
\begin{align*}
\dot{y}_i=&\frac{1}{2}(q_i^e)^T\omega_i^e=-\frac{l_q}{2}\|q_i^e\|^2+\frac{1}{2}(q_i^e)^Tr_i\notag\\
\leq&-\frac{l_q}{4}\|q_i^e\|^2+\frac{1}{4}\|r_i\|^2\notag\\
=&\frac{l_q}{4}(-2y_i+y_i^2)+\frac{1}{4}\|r_i\|^2.
\end{align*}
It can be shown that system $\dot{y}=\frac{l_q}{4}(-2y+y^2)$ with respect to $y\in[0,2]$ is asymptotically stable. For system $\dot{y}=\frac{l_q}{4}(-2y+y^2)+\frac{1}{4}\|r_i\|^2$ with respect to $y\in[0,2]$, by using
input-to-state stability theory \citep{Khalil2002}, it follows that $\lim_{t\rightarrow\infty}y(t)=0$ given the fact $\lim_{t\rightarrow\infty}r_i(t)=0$.  According to Comparison Principal \citep{Khalil2002}, $\lim_{t\rightarrow\infty}y_i(t)=0$ is obtained, that is, $\lim_{t\rightarrow\infty}q_i^e(t)\!=\!0$, which, together with $\lim_{t\rightarrow\infty}r_i(t)=0$, further implies that $\lim_{t\rightarrow\infty}\omega_i^e(t)\!=\!0$, $\forall i\in\mathcal{V}$.
\end{proof}

\begin{proposition}
\label{pro:2}
Consider the virtual position error system \eqref{pos_err_sys_1} with Assumptions \ref{assump:pr} and \ref{assump:graph}. If  $\lim_{t\rightarrow\infty}q_i^e(t)=0$, $\lim_{t\rightarrow\infty}\bar{p}_i(t)=0$ and $\lim_{t\rightarrow\infty}\bar{v}_i(t)=0$ are achieved, the developed
command force $u_i$ in \eqref{command_force} guarantees that $\lim_{t\rightarrow\infty}\tilde{p}_i^e(t)\!=\!0$ and $\lim_{t\rightarrow\infty}\tilde{v}_i^e(t)=0$, $\forall i\in\mathcal{V}$.
\end{proposition}
\begin{proof}
It follows from \eqref{Q_to_R} that $(R_i(Q^e_i)-I_3)\hat{e}_3=\varphi_i^\times q_i^e$, where $\varphi_i=[-q_i^{ey},q_i^{ex},\sigma_i^e]^T$. In terms of $\|R_i(Q_i^c)\|=1$, $\|Q^e_i\|=1$ and \eqref{T_bound}, it follows from $\lim_{t\rightarrow\infty}q_i^e(t)=0$ that each $\frac{T_i}{m_i}(R_i(Q_i)-R_i(Q_i^c))\hat{e}_3=\frac{T_i}{m_i}R_i(Q_i^c)(R_i(Q^e_i)-I_3)\hat{e}_3$ converges to zero asymptotically. This, together with $\lim_{t\rightarrow\infty}\bar{p}_i(t)=0$ and $\lim_{t\rightarrow\infty}\bar{v}_i(t)=0$, guarantees that the perturbation items $\hbar_1$ and $\hbar_2$ in the virtual position error system \eqref{pos_err_sys_2} converge to zero asymptotically. Furthermore, it can be shown that system
\begin{subequations}
\begin{align*}
\dot{\tilde{p}}_i^e=~&\tilde{v}_i^e,\\
\dot{\tilde{v}}^e=~&-l_p\tilde{p}_i^e-l_v\tilde{v}_i^e,
\end{align*}
\end{subequations}
is asymptotically stable. Thus, it follows from input-to-state stability theory \citep{Khalil2002} that $\lim_{t\rightarrow\infty}\tilde{p}_i^e(t)=0$ and $\lim_{t\rightarrow\infty}\tilde{v}_i^e(t)=0$, $\forall i\in\mathcal{V}$ given the fact that $\lim_{t\rightarrow\infty}\hbar(t)=0$, where $\hbar=[\hbar_1^T,\hbar_2^T]^T$.
\end{proof}

\begin{proposition}
\label{pro:3}
Consider the auxiliary system \eqref{aux_sys}. If
$\lim_{t\rightarrow\infty}\tilde{p}_i^e(t)=0$ and $\lim_{t\rightarrow\infty}\tilde{v}_i^e(t)=0$ are achieved, then $\lim_{t\rightarrow\infty}\eta_i(t)=0$ and $\lim_{t\rightarrow\infty}\dot{\eta}_i(t)=0$, $\forall i\!\in\!\mathcal{V}$.
\end{proposition}
\begin{proof}
Denote $\varepsilon_i=l_p\tilde{p}_i^e+l_v\tilde{v}^e_i$ for $i\in\mathcal{V}$. It follows from $\lim_{t\rightarrow\infty}\tilde{p}_i^e(t)=0$ and $\lim_{t\rightarrow\infty}\tilde{v}_i^e(t)\!=\!0$ that  $\lim_{t\rightarrow\infty}\varepsilon_i(t)=0$. To this end, there exists a $\mathcal{KL}$-class function $\rho(\|\varepsilon_i(0)\|,t)$ such that $\|\varepsilon_i(t)\|\leq\rho(\|\varepsilon_i(0)\|,t)$. For $i\in\mathcal{V}$, the following Lyapunov function is proposed:
\begin{equation*}
L_i^\eta=\int_0^{\eta_i+\dot{\eta}_i}\tanh(\chi)^T\mathrm{d}\chi
+\int_0^{\dot{\eta}_i}\tanh(\chi)^T\mathrm{d}\chi+\frac{1}{2k_\eta}\|\dot{\eta}_i\|^2.
\end{equation*}
It is trivial to verify that
\begin{equation}
L_i^\eta\geq\frac{1}{2}(\|\tanh(\eta_i+\dot{\eta}_i)\|^2+\|\tanh(\dot{\eta}_i)\|^2+\frac{1}{k_\eta}\|\dot{\eta}_i\|^2).
\end{equation}
The derivative of $L_i^\eta$ along \eqref{aux_sys} satisfies
\begin{align}
\dot{L}^\eta_i=&\tanh(\eta_i+\dot{\eta}_i)^T\dot{\eta}_i+[\tanh(\eta_i+\dot{\eta}_i)+\tanh(\dot{\eta}_i)\notag\\
&+\frac{1}{k_\eta}\dot{\eta}_i]^T[-k_\eta(\tanh(\eta_i+\dot{\eta}_i)+\tanh(\dot{\eta}_i))+\varepsilon_i]\notag\\
=&-k_\eta\|\tanh(\eta_i+\dot{\eta}_i)+\tanh(\dot{\eta}_i)\|^2-\dot{\eta}_i^T\tanh(\dot{\eta}_i)\notag\\
&+[\tanh(\eta_i+\dot{\eta}_i)+\tanh(\dot{\eta}_i)+\frac{1}{k_\eta}\dot{\eta}_i]^T\varepsilon_i\label{dL_eta_1}\\
\leq&-k_\eta\|\tanh(\eta_i+\dot{\eta}_i)+\tanh(\dot{\eta}_i)\|^2-\!\dot{\eta}_i^T\tanh(\dot{\eta}_i)\notag\\
&+2\sqrt{L^\eta_i}\|\varepsilon_i\|\notag\\
\leq&2\sqrt{L^\eta_i}\rho(\|\varepsilon_i(0)\|,0).\label{dL_eta}
\end{align}
Integrating both sides of \eqref{dL_eta}, we obtain that
\begin{equation}
\sqrt{L^\eta_i(t)}-\sqrt{L^\eta_i(0)}\leq\rho(\|\varepsilon_i(0)\|,0)t, ~~\forall t \geq0,
\end{equation}
 which indicates that $L_i^\eta$ cannot escape to infinity in finite time. In addition, it follows from \eqref{dL_eta_1} that $\dot{L}_i^\eta$ satisfies
\begin{align}
\dot{L}_i^\eta\leq&-\tanh(\bar{\eta}_i)^T\Lambda\tanh(\bar{\eta}_i)\notag\\
&+\left[\|\eta_i+\dot{\eta}_i\|+(1+\frac{1}{k_\eta})\|\dot{\eta}_i\|\right]\rho(\|\varepsilon_i(0)\|,t)\notag\\
\leq&-\min\{1,k_\eta\}\|D\tanh(\bar{\eta}_i)\|^2+c_2\|\bar{\eta}_i\|\rho(\|\varepsilon_i(0)\|,t)\notag\\
\leq&-c_1\|\tanh(\bar{\eta}_i)\|^2+c_2\|\bar{\eta}_i\|\rho(\|\varepsilon_i(0)\|,t),
\end{align}
where $\bar{\eta}_i=[\eta_i^T+\dot{\eta}_i^T,\dot{\eta}_i^T]^T$, $\Lambda=I_{3}\otimes\left[\begin{array}{cc}k_\eta&k_\eta\\k_\eta&k_\eta+1\end{array}\right]$,
$D=I_{3}\otimes\left[\begin{array}{cc}1&1\\1&2\end{array}\right]$, $c_1=\frac{3-\sqrt{5}}{2}\min\{1,k_\eta\}$ and $c_2=1+\frac{1}{k_\eta}$. Since $L^\eta_i$ cannot escape in finite time, there exist $t_1$ and $\Delta_\eta$ such that $\|\bar{\eta}_i(t)\|\leq\Delta_\eta$ for $ t\in[0,t_1]$ and $\rho(\|\varepsilon_i(t_1)\|,t-t_1)\leq\frac{c_1(\chi_\eta)^2\Delta_\eta}{c_2}$ for $t\in[t_1,\infty)$, where $\chi_\eta<\frac{\tanh(\Delta_\eta)}{\Delta_\eta}$ is a constant. In particular, $\|\tanh(\bar{\eta}_i)\|\leq\chi_\eta\|\bar{\eta}_i\|$. Thus, for $t\geq t_1$, $\dot{L}^\eta_i$ satisfies
\begin{equation}
\dot{L}_i^\eta\leq-c_1(\chi_\eta)^2\|\bar{\eta}_i\|^2+c_2\|\bar{\eta}_i\|\rho(\|\varepsilon_i(t_1)\|,t-t_1).
\end{equation}
This implies that $\dot{L}^\eta_i$ is negative outside the set
\begin{equation*}
\mathcal{Z}_i=\left\{\bar{\eta}_i\in\mathbb{R}^6\mid\|\bar{\eta}_i\|\leq
\frac{c_2}{c_1(\chi_\eta)^2}\rho(\|\varepsilon_i(t_1)\|,t-t_1)\right\}.
\end{equation*}
Thus, $\bar{\eta}_i$ is bounded and ultimately converges to the set $\mathcal{Z}_i$. In view of $\lim_{t\rightarrow\infty}\rho(\|\varepsilon_i(t_1)\|,t-t_1)=0$, it follows that $\lim_{t\rightarrow\infty}\bar{\eta}_i(t)=0$, which implies that $\lim_{t\rightarrow\infty}\eta_i(t)\!=\!0$ and $\lim_{t\rightarrow\infty}\dot{\eta}_i(t)\!=\!0$, $\forall i\in\mathcal{V}$.
\end{proof}
Since Th2.i)-Th2.iii) have been proven, it can be concluded that the coordinated trajectory tracking of multiple VTOL UAVs is achieved in the sense of \eqref{objective}.
\end{proof}

\section{Simulations}
\label{sec:4}
\begin{figure}[t]
  \centering
  \includegraphics[width=2in]{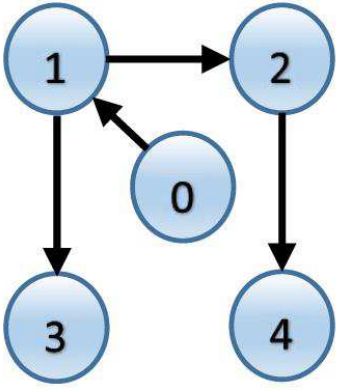}
  \caption{Leader-follower graph $\mathcal{G}_{n+1}$.}
  \label{fig:graph}
\end{figure}
In this section, simulations are performed to verify the proposed distributed control approach on a formation of four VTOL UAVs described by \eqref{pos_kin}-\eqref{att_dyn}. The inertial parameters are assumed to be identical: $m_i=0.85(\mathrm{kg})$ and $J_i=\mathrm{diag}\{4.856,4.856,8.801\}\times10^{-2}(\mathrm{kgm^2})$, $i=1,2,3,4$.
The leader-follower graph $\mathcal{G}_{n+1}$ is illustrated in Fig. \ref{fig:graph}, where each arrow denotes the corresponding information flow. Furthermore, define $d_{ij}=1$ if follower $j$ is accessible to follower $i$, and $d_{ij}=0$, otherwise, for $i,j=0,1,\cdots,4$. The desired trajectory is described as $p_r(t)=[5\cos(0.2t),5\sin(0.2t),t]^T(\mathrm{m})$, and the desired position offsets of the followers relative to the leader are $\delta_1=[2,2,0]^T(\mathrm{m})$, $\delta_2=[2,-2,0]^T(\mathrm{m})$, $\delta_3=[-2,-2,0]^T(\mathrm{m})$ and $\delta_4 = [-2,2,0]^T(\mathrm{m})$, respectively. This indicates that the desired formation pattern is a square. The distributed estimator states of each follower UAV are initialized as zeros. The follower UAVs are initially situated at $p_1(0)=[5,3,-1]T(\mathrm{m})$, $p_2(0)=[9,-4,1]^T(\mathrm{m})$, $p_3(0)=[4,-2,-3]^T(\mathrm{m})$ and $p_4(0)=[-1,4,-2]^T(\mathrm{m})$ and their initial attitudes are $Q_i(0)=[1,0,0,0]^T$, $i=1,2,3,4$. The estimator and control parameters are chosen as
follows: $k_\gamma=0.5$, $k_p=k_v=8$ based on \eqref{k_pv}, $l_a=12$ based on \eqref{l_a}, $k_a=4$ based on \eqref{beta}, $l_p=l_v=k_\eta=4$ and $l_q=k_q=16$. The simulation results are illustrated in Figs. \ref{fig:3D}-\ref{fig:vel_err}.

\begin{figure}[t]
  \centering
  \includegraphics[width=3in]{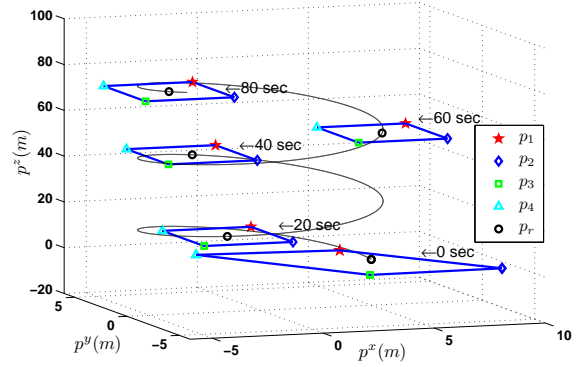}
  \caption{Snapshots of coordinated trajectory tracking of VTOL UAVs.}
  \label{fig:3D}
\end{figure}

\begin{figure}[t]
  \centering
  \includegraphics[width=3in]{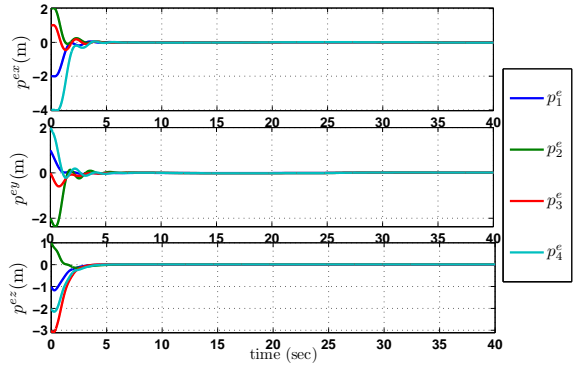}
  \caption{Position error of follower VTOL UAVs.}
  \label{fig:pos_err}
\end{figure}

\begin{figure}[t]
  \centering
  \includegraphics[width=3in]{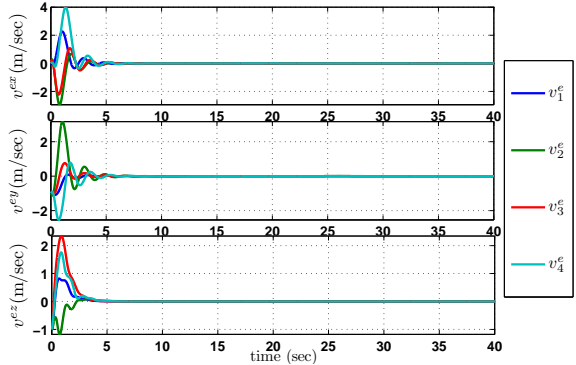}
  \caption{Velocity error of follower VTOL UAVs.}
  \label{fig:vel_err}
\end{figure}

Fig. \ref{fig:3D} exhibits the evolution of the VTOL UAV formation with respect to the leader in the three-dimensional space, where the formation is depicted every twenty seconds. It can be seen that the follower UAVs
reach the prescribed square pattern while tracking the leader. Figs. \ref{fig:pos_err} and \ref{fig:vel_err} describe the position and velocity errors of the follower UAVs with respect to the leader. It can be observed that each error component converges to zero asymptotically.
Consequently, the simulation results validate that the proposed distributed control approach effectively guarantees the coordinated trajectory tracking of multiple VTOL UAVs in the sense of \eqref{objective}.

\section{Conclusion}
\label{sec:5}
A distributed control  strategy is proposed in this paper to achieve the coordinated trajectory tracking of multiple VTOL UAVs with local information exchange. The connectivity of the network graph is weak in the sense that we only require the graph to contain a directed spanning tree. A novel distributed estimator is firstly designed for each VTOL UAV to obtain the leader's desired information asymptotically. Then, under the hierarchical framework, a command force and an applied torque are exploited for each VTOL UAV to fulfill the accurate tracking to the desired information asymptotically. In addition, an auxiliary system is introduced in the control development to avoid the non-singular command attitude extraction and the use of the unavailable desired information. Simulations are carried out to validate the theoretical results.


\bibliographystyle{IEEEtran}
\bibliography{refs}

\end{document}